\documentclass[draftcls,onecolumn,11pt]{IEEEtran}

\usepackage{amsmath,amstext,amssymb,epsf}
\usepackage{fullpage,latexsym}
\usepackage{epsfig}
\usepackage{setspace}
\numberwithin{equation}{section} 

 \newtheorem{lemma}{Lemma}[section]
 \newtheorem{theorem}[lemma]{Theorem}

 \newtheorem{claim}[lemma]{Claim}
 \newtheorem{corollary}[lemma]{Corollary}

 \newtheorem{rem}[lemma]{Remark}
\newenvironment{remark}{\begin{rem}}{\hspace*{\fill}$\diamondsuit$\end{rem}}
 \newtheorem{ex}[lemma]{Example}

\begin{document}

\title{Normalized Information Distance is Not Semicomputable}
\author{Sebastiaan A. Terwijn\thanks{SAT is with the Radboud University Nijmegen;
Email terwijn@logic.at}, Leen Torenvliet\thanks{LT is with the University 
of Amsterdam; Email leen@science.uva.nl}, and Paul M.B. Vit\'{a}nyi\thanks{
PMBV is with the CWI and the University of Amsterdam; Email Paul.Vitanyi@cwi.nl}}
\maketitle

\begin{abstract}
Normalized information distance (NID) 
uses the theoretical notion of Kolmogorov complexity, which
for practical purposes is approximated by the length
of the compressed version of the file involved, using
a real-world compression program. This practical application is 
called `normalized compression distance' and it is trivially computable.
It is a parameter-free similarity measure
based on compression, and is used in pattern recognition, data mining,
phylogeny, clustering, and classification.
The complexity properties of its theoretical precursor, the NID, have been
open. 
We show that the NID is neither upper
semicomputable nor lower semicomputable.

{\em Index Terms}---
Normalized information distance, 
Kolmogorov complexity,
semicomputability.
\end{abstract}

\section{Introduction}
\label{sect.intro}

The classical notion of
Kolmogorov complexity \cite{Ko65} is an objective measure
for the information in 
a {\em single} object, and information distance measures the information
between a {\em pair} of objects \cite{BGLVZ98}. This last notion has
spawned research in the theoretical direction, among others
\cite{CMRSV02,VV02,Vy02,Vy03,MV01,SV02}.
Research in the practical direction has focused on
the normalized information distance (NID),
also called the similarity metric, which arises
by normalizing the information distance in a proper manner. 
(The NID is defined by \eqref{eq.nid} below.) If we also
approximate the Kolmogorov complexity through real-world
compressors \cite{Li03,CVW03,CV04}, 
then we obtain the
normalized compression distance (NCD). 
This is a parameter-free, feature-free, and alignment-free
similarity measure  that
has had great
impact in applications.
The NCD was preceded by a related 
nonoptimal distance \cite{LBCKKZ01}.
In \cite{KLRWHL07} another variant of the NCD has been tested
on all major time-sequence databases used in all major data-mining conferences
against all other major methods used. The compression method turned out
to be competitive in general 
and superior in
heterogeneous data clustering and anomaly detection.
There have been many applications in pattern recognition, phylogeny, clustering,
and classification, ranging from
hurricane forecasting and music to
to genomics and analysis of network traffic,
see the many 
papers referencing
\cite{Li03,CVW03,CV04}
in Google Scholar. 
The NCD is trivially computable.
In \cite{Li03} it is shown that its theoretical precursor, the NID, is 
a metric up to negligible discrepancies
in the metric (in)equalities and that it is always between 0 and 1.
(For the subsequent computability notions see
Section~\ref{sect.prel}.)

The computability status of the NID has been open,
see Remark VI.1 in \cite{Li03} which asks whether the NID is 
upper semicomputable, and (open) Exercise 8.4.4 (c) in the textbook
\cite{LiVi08} which asks whether the NID is semicomputable at all. 
We resolve this question by showing the following.
\begin{theorem}
Let $x,y$ be strings and denote the NID between them 
by $e(x,y)$.

(i) The function $e$ is not lower semicomputable 
(Lemma~\ref{lem.lower}).

(ii) The function $e$ is not upper semicomputable 
(Lemma~\ref{lem.upper}).
\end{theorem}

Item (i) implies that there is no pair of lower semicomputable
functions $g,\delta$ such that
$g(x,y) + \delta(x,y)=e(x,y)$. (If there were such a pair, then $e$ 
itself would be lower semicomputable.) Similarly, Item (ii) implies that
there is no pair of upper semicomputable
functions $g,\delta$ such that
$g(x,y) + \delta(x,y)=e(x,y)$. Therefore,
the theorem implies 
\begin{corollary}
(i) The NID $e(x,y)$ cannot be approximated by
a semicomputable function $g(x,y)$ to any
computable precision $\delta (x,y)$. 

(ii) The NID $e(x,y)$ cannot be approximated by
a computable function $g(x,y)$ to any
semicomputable precision $\delta (x,y)$.
\end{corollary}
How can this be reconciled with the above applicability
of the NCD (an approximation of the NID through real-world compressors)?
It can be speculated upon but not proven that natural data do not contain
complex mathematical regularities such as $\pi=3.1415 \ldots$ or
a universal Turing machine computation. The
regularities they do contain are of the sort detected by a good compressor.
In this view, the Kolmogorov complexity and the length of
the result of a good compressor are not that different for natural data.

\section{Preliminaries}\label{sect.prel}

We write {\em string} to mean a finite binary string,
and $\epsilon$ denotes the empty string.
The {\em length} of a string $x$ (the number of bits in it)
is denoted by $|x|$. Thus,
$|\epsilon| = 0$.
Moreover, we identify strings with natural numbers  
by associating each string with its index
in the length-increasing lexicographic ordering
\[
( \epsilon , 0),  (0,1),  (1,2), (00,3), (01,4), (10,5), (11,6), 
\ldots . 
\]
Informally, the Kolmogorov complexity of a string
is the length of the shortest string from which the original string
can be losslessly reconstructed by an effective
general-purpose computer such as a particular universal Turing machine $U$,
\cite{Ko65}.
Hence it constitutes a lower bound on how far a
lossless compression program can compress.
In this paper we require that the set of programs of $U$ is prefix free
(no program is a proper prefix of another program), that is, we deal
with the {\em prefix Kolmogorov complexity}. 
(But for the results in this paper it does not matter whether we use
the plain Kolmogorov complexity or the prefix Kolmogorov complexity.)
We call $U$ the {\em reference universal Turing machine}.
Formally, the {\em conditional prefix Kolmogorov complexity}
$K(x|y)$ is the length of the shortest input $z$
such that the reference universal Turing machine $U$ on input $z$ with
auxiliary information $y$ outputs $x$. The
{\em unconditional prefix Kolmogorov complexity} $K(x)$ is defined by
$K(x|\epsilon)$.
For an introduction to the definitions and notions
of Kolmogorov complexity (algorithmic information theory)
see ~\cite{LiVi08}.

Let $\cal N$ and $\cal R$ denote
the nonnegative integers and the real numbers,
respectively. 
A function $f: {\cal N} \rightarrow {\cal R}$ is
{\em upper semicomputable} (or $\Pi_1^0$)
if it is defined by a rational-valued computable
function  $\phi (x,k)$ where $x$ is a string
and $k$ is a nonnegative integer
such that $\phi(x,k+1) \leq \phi(x,k)$ for every $k$ and
  $\lim_{k \rightarrow \infty} \phi (x,k)=f(x)$.
This means
  that $f$ can be computably approximated from above.
A function $f$ is
{\em lower semicomputable} (or $\Sigma_1^0$) 
  if $-f$ is upper semicomputable.
 A function is called
{\em semicomputable} (or $\Pi_1^0 \bigcup \Sigma_1^0$)
  if it is either upper semicomputable or lower semicomputable or both.
A function
$f$ is {\em computable}
(or recursive)
  iff it is both upper semicomputable and
lower semicomputable (or $\Pi_1^0 \bigcap \Sigma_1^0$).
Use $ \langle \cdot \rangle$
as a {\em pairing
function}
over ${\cal N}$ to associate a unique natural number $\langle x, y \rangle$
with each pair $(x, y)$ of natural numbers.
An example is $\langle x, y \rangle$ defined by
$y+(x+y+1)(x+y)/2$.
In this way we can extend the above definitions to functions of two
nonnegative integers, in particular to distance functions.

The {\em information distance} $D(x,y)$ between strings $x$ and $y$
is defined as 
\[
D(x,y)= \min_{p} \{|p|: U(p,x)=y \wedge U(p,y)=x \},
\]
where $U$ is the reference universal Turing machine above.
Like the Kolmogorov complexity $K$, the distance function $D$ 
is upper semicomputable. Define
\[E(x,y)= \max \{K(x|y),K(y|x)\}.
\]
In \cite{BGLVZ98} it is shown that
the function $E$ is upper semicomputable, 
$D(x,y)= E(x,y)+O(\log E(x,y))$, the function $E$ is a metric (more precisely,
that it satisfies the metric (in)equalities up to a constant),
and that $E$ is minimal (up to a constant) among all 
upper semicomputable distance functions $D'$ satisfying the mild
normalization conditions $\sum_{y:y \neq x} 2^{-D'(x,y)} \leq 1$ and
$\sum_{x:x \neq y} 2^{-D'(x,y)} \leq 1$. 
(Here and elsewhere in this paper ``$\log$'' denotes the binary logarithm.)
It should be mentioned
 that the minimality property
was relaxed from the $D'$ functions being metrics \cite{BGLVZ98} 
to symmetric distances
\cite{Li03} to the present form \cite{LiVi08} without serious
proof changes.
The {\em normalized information distance} (NID) $e$ is defined by
\begin{equation}\label{eq.nid}
e(x,y) = \frac{E(x,y)}{\max\{K(x),K(y)\}}.
\end{equation}
It is straightforward that $0 \leq e(x,y) \leq 1$ up to some minor
discrepancies for all $x,y \in \{0,1\}^*$.
Since $e$ is the ratio between two upper semicomputable functions,
that is, between two $\Pi_1^0$ functions, it is a $\Delta_2^0$ function.
That is, $e$  is computable relative to the halting problem $\emptyset '$.
One would not expect any better bound in the arithmetic hierarchy.
However, we can say this: Call a function $f(x,y)$
{\em computable in the limit} if there exists
a rational-valued computable function $g(x,y,t)$ such that
 $\lim_{t \rightarrow \infty} g(x,y,t)$ $=f(x,y)$.
This is precisely the class of functions
that are Turing-reducible
to the halting set, and
the  NID
is in this class, Exercise~8.4.4 (b) in \cite{LiVi08} (a result due to
\cite{Ga01}).

In the sequel we use time-bounded Kolmogorov complexity.
Let $x$ be a string of length $n$ and $t(n)$ 
a computable time bound. Then $K^t$ denotes the {\em time-bounded} version
of $K$ defined by
\[
K^t(x|y) = \min_p \{|p|: U'(p,y) = x \; \text{\rm in at most} \; t(n) \; 
\text{\rm steps}\}.
\]
Here we use the two work-tape reference universal Turing machine $U'$
suitable
for time-bounded Kolmogorov complexity \cite{LiVi08}.
The computation of $U'$ is measured in terms of the output rather than
the input, which is more natural in the context of Kolmogorov complexity.

\section{The NID is not lower semicomputable}

Define the
time-bounded version $E^t$ of $E$ by
\begin{equation}\label{eq.Et}
E^t(x,y) = \max \{K^t(x|y), K^t(y|x) \}.
\end{equation}
\begin{lemma}\label{lem.1}
For every length $n$ and computable time bound $t$ there are 
strings $u$ and $v$ of length $n$ such that
\begin{itemize}
\item $K(v) \geq n-c_1$,
\item $K(v|u) \geq n-c_2$,
\item $K(u|n) \leq c_2$,
\item $K^t(u|v) \geq n-c_1 \log n - c_2$,
\end{itemize}
where $c_1$ 
is a nonnegative constant independent
of $t,n$, and $c_2$
is a nonnegative constant depending on $t$
but not on $n$.
\end{lemma}

\begin{proof}
Fix an integer $n$.
There is a $v$ of length $n$ such that
$K(v|n) \geq n$ by simple counting (there are $2^n$ strings of length $n$
and at most $2^n-1$ programs of length less than $n$). 
If we have a program for $v$ then we can
turn it into a program for $v$ ignoring conditional information
by adding a constant number of bits. Hence, $K(v)+c \geq K(v|n)$ for some
nonnegative constant $c$. Therefore, for large enough nonnegative 
constant $c_1$ we have
\[
K(v) \geq n-c_1.
\] 
Let $t$ be a computable time bound and let the computable
time bound $t'$ be large enough with respect to $t$
so that the arguments
below hold. 
Use the reference universal Turing machine $U'$ with input $n$
 to run all programs of length
less than $n$ for $t'(n)$ steps. Take the least string 
$u$ of length $n$ not occurring
as an output among the halting programs. Since there are at
most $2^n-1$ programs as above, and $2^n$ strings of length $n$ there is
always such a string $u$. 
By construction $K^{t'}(u|n) \geq n$ and for a large enough 
constant $c_2$ also
\[
K(u|n) \leq c_2, 
\]
where $c_2$ depends on $t'$ (hence $t$) but not on $n,u$.
Since $u$ in the conditional only supplies $c_2$ bits
apart from its length $n$ we have 
\[
K(v|u) \geq K(v|n)-K(u|n) \geq n-c_2.
\]
This implies also that $K^{t'} (v|u) \geq n-c_2$. Hence,
\[
2n-c_2 \leq K^{t'}(u|n) + K^{t'} (v|u).
\]
Now we use the
time-bounded symmetry of algorithmic information 
\cite{Lo86} (see also 
\cite{LiVi08}, Exercise 7.1.12) where $t$ is given and  $t'$ is choosen
in the standard proof of
the symmetry of algorithmic information \cite{LiVi08}, Section 2.8.2
(the original is due to L.A. Levin and A.N. Kolmogorov
in \cite{ZL70}), so that the statements below hold.
(Recall also
that for large enough $f$, 
$K^f(v|u,n)=K^f(v|u)$ and $K^f(u|v,n)=K^f(u|v)$ since in the
original formulas $n$ is
present in each term.) Then,
\[
K^{t'}(u|n) + K^{t'} (v|u) - c_1 \log n 
\leq  K^{t'}(v,u|n), 
\]
with the constant $c_1$ large enough and independent of
$t,t',n,u,v$.
For an appropriate choice of $t'$ with respect to $t$ it is easy to see 
(the simple side of the time-bounded symmetry of algorithmic information)
that
\[ 
K^{t'}(v,u|n) 
\leq  K^t(v|n)+K^t(u|v).
\]
Since $K^t(v|n) \geq K(v|n)  \geq n$ we obtain 
$K^t(u|v) \geq n - c_1 \log n - c_2$.
\end{proof}

A similar but tighter result can be obtained from  \cite{BT09}, Lemma 7.7.

\begin{lemma}\label{lem.xor}
For every length $n$ and computable time bound $t$ (provided $t(n) \geq cn$ for
a large enough constant $c$),
there exist strings $v$ and $w$ of length $n$ such that
\begin{itemize}
\item $K(v) \geq n-c_1$,
\item $E(v,w) \leq c_3$,
\item $E^t(v,w) \geq n- c_1 \log n -c_3$, 
\end{itemize}
where the nonnegative constant $c_3$  depends on $t$
but not on $n$ and the nonnegative constant $c_1$ is independent
of $t,n$.
\end{lemma} 
\begin{proof}
Let strings $u,v$ and constants $c_1,c_2$ be as in Lemma~\ref{lem.1} 
using $2t$ instead of $t$,
and the constants $c',c'',c_3$ are large enough for the
proof below.
By Lemma~\ref{lem.1}, we have $K^{2t} (u|v) \geq n-c_1 \log n - c_2$
with $c_2$ appropriate for the time bound $2t$. 
Define $w$ by $w=v \oplus u$
where $\oplus$ denotes the bitwise XOR. Then,
\[
E(v,w) \leq K(u|n)+c' \leq c_3,
\]  
where the nonnegative constant $c_3$ depends on $2t$ (since $u$ does)
but not on $n$ and the constant $c'$ is independent of $t,n$.
We also have $u=v \oplus w$ so that (with the time bound $t(n)\geq cn$ for
$c$ a large enough constant independent of $t,n$)
\begin{eqnarray*}
n-c_1 \log n - c_2 & \leq & K^{2t} (u|v) 
\\& \leq & K^t(w|v) +c'
\\&\leq& \max \{K^t(v|w), K^t(w|v)\}+c''
\\&=& E^t(v,w)+c'',
\end{eqnarray*} 
where the nonnegative constants $c',c''$ are independent of $t,n$.
\end{proof}

\begin{lemma}\label{lem.lower}
The function $e$ is not lower semicomputable.
\end{lemma}

\begin{proof}
Assume by way of contradiction that the lemma is false.
Let $e_i$ be a lower semicomputable
function approximation of $e$ such that 
$e_{i+1}(x,y) \geq e_i(x,y)$ for
all $i$ and $\lim_{i \rightarrow \infty} e_i(x,y) = e(x,y)$. 
Let $E_i$ be an upper semicomputable function approximating $E$ such
that $E_{i+1}(x,y) \leq E_i(x,y)$ for all $i$ 
and $\lim_{i \rightarrow \infty} E_i(x,y) = E(x,y)$.
Finally, for $x,y$ are strings of length $n$ let  $i_{x,y}$ denote
the least $i$ such that
\begin{equation}\label{eq.eE}
e_{i_{x,y}}(x,y)  \geq \frac{E_{i_{x,y}}(x,y)}{n+2 \log n +c} \;,
\end{equation}
where $c$ is a large enough constant (independent of $n,i$) such that
$K(z) <  n+2 \log n +c$ for every string $z$ of length $n$
(this follows from the upper bound on $K$, see \cite{LiVi08}). 
Since the function $E$ is upper semicomputable and
the function $e$ is 
lower semicomputable by the contradictory assumption
such an ${i_{x,y}}$ exists.
Define the function $s$ by $s(n)=\max_{x,y \in \{0,1\}^n} \{i_{x,y}\}$.
\begin{claim}\label{claim.es}
The function $s(n)$ is total computable
and 
$E^s (v,w))\geq n-c_1 \log n-c_3$ for some strings $v,w$
of length $n$ and constants $c_1,c_3$ in Lemma~\ref{lem.xor}.
\end{claim}
\begin{proof}
By the contradictory assumption $e$ is lower semicomputable, and
$E$ is upper semicomputable since $K( \cdot | \cdot)$ is. 
Recall also that $e(x,y) > E(x,y)/(n+2 \log n +c)$ for every pair
$x,y$ of strings of length $n$.
Hence for every such pair $(x,y)$ we can compute
$i_{x,y} < \infty$. Since $s(n)$ is the maximum of $2^{2n}$ computable integers,
$s(n)$ is computable as well and total.
Then, the claim follows from Lemma~\ref{lem.xor}.
(If $s(n)$ happens to be too small to apply Lemma~\ref{lem.xor} 
we increase it total computably until it is large enough.)
\end{proof}
\begin{remark}
The string $v$ of length $n$ as defined in the proof of
Lemma~\ref{lem.1} satisfies $K(v|n) \geq n$. Hence $v$
is incomputable \cite{LiVi08}. Similarly this holds for $w=v \oplus u$
(defined in Lemma~\ref{lem.xor}). But above we look for a
function $s(n)$ such that {\em all} pairs $x,y$ of strings of length $n$
(including the incomputable strings $v,w$)
satisfy \eqref{eq.eE} with $s(n)$ replacing
$i_{x,y}$. Since the computable function $s(n)$ does
not depend on the particular strings $x,y$ but only on their
length $n$, we can use it as the computable time bound $t$ in 
Lemmas~\ref{lem.1} and \ref{lem.xor}
to define strings $u,v,w$ of length $n$.

For given strings $x,y$ of length $n$, the value
$E_{i_{x,y}}(x,y)$ is not necessarily equal to $E^s(x,y)$. 
Since $s(n)$ majorises
the $i_{x,y}$'s 
and $E$
is upper semicomputable, we have
$E^s (x,y) \leq E_{i_{x,y}}(x,y)$,
for all pairs $(x,y)$ of strings $x,y$ of length $n$. 
\end{remark}
Since $K(v) \geq n-c_1$
we have  $E(v,w) \geq e(v,w)  (n-c_1) $. 
By the contradictory assumption that $e$
is lower semicomputable we have $e(v,w) \geq 
e^s (v,w)$.
By \eqref{eq.eE} and the definition of $s(n)$
we have
\[
e^s(v,w) \geq  
\frac{E^s(v,w) }{n+2 \log n +c} \;.
\]
Hence, 
\[
E(v,w) \geq \frac{E^s(v,w) (n-c_1)}{n+2 \log n +c} \;.
\] 
But $E(v,w) \leq c_3$ by Lemma~\ref{lem.xor} and 
$E^s(v,w) \geq n-c_1 \log n-c_3$ 
by Claim~\ref{claim.es}, which yields the required contradiction
for large enough $n$.
\end{proof}

\section{The NID is not upper semicomputable}

\begin{lemma}\label{lem.upper}
The function $e$ is not upper semicomputable.
\end{lemma}

\begin{proof}
\rm
It is easy to show that $e(x,x)$ (and hence $e(x,y)$
in general) is not upper semicomputable. 
For simplicity
we use $e(x,x) = 1/K(x)$.
Assume that the function
$1/K(x)$ is upper semicomputable 
Then, $K(x)$ is lower semicomputable.
Since $K(x)$ is also upper semicomputable, it
is computable. But this violates the known fact \cite{LiVi08}
that $K(x)$ is incomputable.
\end{proof}

\section{Open Problem}

A subset of $\mathcal{N}$ is
called $n$-computably enumerable ($n$-c.e.) if it is 
a Boolean combination of $n$ computably enumerable
sets. Thus, the $1$-c.e. sets are the computably enumerable
sets, the $2$-c.e. sets (also called d.c.e.) the differences
of two c.e. sets, and so on.
The $n$-c.e. sets are referred to as the
{\em difference hierarchy} over the c.e. sets.
This is an effective analog of a classical hierarchy
from descriptive set theory.
Note that a set is $n$-c.e. if it has a
computable approximation that changes at most $n$ times.

We can extend the notion of $n$-c.e. set to a notion that
measures the number of fluctuations of a function as follows:
For every $n\geq 1$,
call $f:\mathcal{N}\rightarrow \mathcal{R}$
{\em $n$-approximable} if there is a rational-valued
computable approximation $\phi$ such that
$\lim_{k\rightarrow\infty} \phi(x,k) = f(x)$ and such that for every $x$,
the number of $k$'s such that
$\phi(x,k+1) - \phi(x,k)<0$ is bounded by~$n-1$.
That is, $n-1$ is a bound on the number of fluctuations
of the approximation.
Note that the $1$-approximable functions are precisely the
lower semicomputable ($\Sigma^0_1$) ones (zero fluctuations).
Also note that a set $A\subseteq\mathcal{N}$ is $n$-c.e.
if and only if the characteristic function of $A$
is $n$-approximable.

{\bf Conjecture} For every $n\geq 1$, the normalized information
distance $e$ is not $n$-approximable.

\section{Acknowledgement}
Harry Buhrman 
pointed out an apparent circularity in an early version of the proof
of Claim~\ref{claim.es}. A referee pointed out
that an early version of Lemma~\ref{lem.upper} was incorrect and
gave many other useful comments that improved the paper.

\end{document}